\documentclass[a4paper,11pt]{article}
\pdfoutput=1
\usepackage{latexsym}
\usepackage{amsmath, amsthm, amssymb, bbm}
\usepackage[mathscr]{eucal}
\usepackage{tikz}
\usepackage{cite}
\usepackage{hyperref}
\usepackage{graphicx}

\theoremstyle{plain}
\newtheorem{proposition}{Proposition}

\newtheorem{lemma}[proposition]{Lemma}

\theoremstyle{definition}

\theoremstyle{remark}
\newtheorem{remark}[proposition]{Remark}

\newcommand{\rhs}{r.h.s.\ }
\newcommand{\lhs}{l.h.s.\ }
\newcommand{\wrt}{w.r.t.\ }
\newcommand{\cf}{cf.\ }

\newcommand{\ud}{\mathrm{d}}
\newcommand{\del}{\partial}

\DeclareMathOperator{\supp}{supp}

\DeclareMathOperator{\dom}{dom}

\newcommand{\betrag}[1]{{\lvert #1 \rvert}}
\newcommand{\norm}[1]{{\lVert #1 \rVert}}
\newcommand{\R}{\mathbb{R}}
\newcommand{\C}{\mathbb{C}}
\newcommand{\Z}{\mathbb{Z}}

\newcommand{\D}{\mathfrak{D}}
\newcommand{\skal}[2]{\langle #1 , #2 \rangle}

\newcommand{\order}{\mathcal{O}}

\newcommand{\HS}{\mathcal{H}}
\newcommand{\sS}{\mathcal{S}}
\newcommand{\sK}{\mathcal{K}}
\newcommand{\sF}{\mathcal{F}}
\newcommand{\sD}{\mathcal{D}}

\newcommand{\eps}{\varepsilon}

\newcommand{\N}{\mathbb{N}}

\newcommand{\bk}{{\mathrm{bulk}}}
\newcommand{\bd}{{\mathrm{boundary}}}

\newcommand{\vp}{{\varphi}}

\newcommand{\beq}{\begin{equation}}
\newcommand{\eeq}{\end{equation}}

\begin{document}


\title{Generalized Wentzell boundary conditions and quantum field theory}
\author{Jochen Zahn \\ Institut f\"ur Theoretische Physik, Universit\"at Leipzig \\ Br\"uderstr.~16, 04103 Leipzig, Germany \\ jochen.zahn@itp.uni-leipzig.de}

\date{\today}

\date{\today}


\maketitle

\begin{abstract}
We discuss a free scalar field subject to generalized Wentzell boundary conditions. On the classical level, we prove well-posedness of the Cauchy problem and in particular causality. Upon quantization, we obtain a field that may naturally be restricted to the boundary. We discuss the holographic relation between this boundary field and the bulk field.
\end{abstract}

\section{Introduction}

Holography has been a main theme in theoretical high energy physics and quantum gravity in the last two decades \cite{tHooftHolography, SusskindHolography, MaldacenaHolography, WittenHolography}. Inspired by the gauge/gravity duality, studies of holographic aspects were often considering $d+1$ dimensional Anti deSitter space (AdS) and its (conformal) boundary, $d$ dimensional Minkowski space. However, holography is a generic aspect of quantum field theory on space-times with time-like boundaries, raising the question which of the properties of holography on AdS are generic, and which ones are special to AdS.

By holography we here mean an isomorphism of the algebras of bulk and boundary observables. For AdS such an isomorphism was abstractly constructed in \cite{RehrenHolography}, using the fact that wedges in AdS have double cones as boundaries: one identifies the algebras of observables localized in wedges of AdS with the algebras of observables localized in the corresponding doubles cones on the boundary. One can also proceed via Wightman functions, and define the Wightman function of the boundary theory to coincide with appropriately scaled limits of bulk Wightman functions \cite{BertolaBrosMoschellaSchaeffer00}. This amounts to defining the boundary field as an appropriately scaled limit of the bulk field, \cf also \cite{DutschRehrenDualField} for the relation between the boundary value of the bulk field and the ``dual field'' in the original formulation of the AdS/CFT correspondence. Let us list some of the properties of holography on AdS:
\begin{itemize}
\item The boundary theory is a conformal field theory. This is a generic feature on AdS \cite{RehrenHolography, BertolaBrosMoschellaSchaeffer00}, independent of the concrete choice of the bulk fields.
\item The correspondence maps bulk observables localized in a compact region to boundary observables localized in compact regions.
\item For a bulk theory with local observables, the boundary theory will not fulfill the time-slice axiom \cite{RehrenHolography}.
\item The boundary conformal field in general has a positive anomalous dimension. The basic example is the massive scalar field \cite{WittenHolography}, where the anomalous dimension $((\frac{d}{2})^2 + \mu^2)^{\frac{1}{2}} + 1$ of the dual field is strictly positive.
\end{itemize}

In the following, we study the holographic relation between a massive scalar field on $d+1$ dimensional Minkowski space with $d$ dimensional time-like boundaries. Not surprisingly, we find that the boundary field theory is not conformal, and that the bulk observables localized in compact space-time regions are mapped to boundary observables that are delocalized. The first two properties in the above list thus seems to be specific to AdS. Also in our setting, the time-slice axiom does not hold for the boundary theory.

Regarding the last point in the above list, it seems obvious that the boundary field, being the boundary limit of the bulk field, inherits its short-distance behavior. Hence, for a scalar field one would expect a short-distance singularity $(x - y)^{- (d-1)}$ for the two-point function. However, it turns out that there are boundary conditions ensuring that the boundary two-point function has short-distance singularity $(x - y)^{- (d-2)}$, as one expects for a scalar field in $d$ space-time dimensions. These are so-called generalized Wentzell boundary conditions.
These are of interest in their own right, and the first part of this work will be devoted to their study.

Concretely, we consider a free scalar field on the $(d+1)$-dimensional space-time $M = \R \times \Sigma$, 
with the spatial slices $\Sigma$ having a boundary $\del \Sigma$. The two main examples will be the half-space $\Sigma = \R^d_+ =\R^{d-1} \times [0, \infty)$ and the strip $\Sigma = \R^{d-1} \times [-S, S]$.
We do not impose boundary conditions by hand, but supplement the bulk action with an action for the boundary, which is of the same form. Concretely,
\begin{align}
 \mathcal{S} & = \mathcal{S}_\bk + \mathcal{S}_\bd \nonumber \\
\label{eq:action}
 & = - \frac{1}{2} \int_{M} g^{\mu \nu} \del_\mu \phi \del_\nu \phi + \mu^2 \phi^2 - \frac{c}{2} \int_{\del M} h^{\alpha \beta} \del_\alpha \phi \del_\beta \phi + \mu^2 \phi^2
\end{align}
where $g$ is the Minkowski metric on the bulk and $h$ is the induced metric on the boundary. $c$ is a positive constant with the dimension of a length. The simplest case where actions of this type appear may be the open Nambu-Goto string with masses at the ends \cite{ChodosThorn74}, i.e.,
\[
 \mathcal{S} = - \gamma \int_\Xi \sqrt{\betrag{g}} - m \int_{\del \Xi} \sqrt{\betrag{h}},
\]
where $\Xi$ is the world-sheet, $g$ the induced metric in the bulk and $h$ the induced metric on the boundary. In fact, for rotating string solutions to this action, the quadratic part of the action for fluctuations normal to the plane of rotation are exactly of the form \eqref{eq:action} with $d=1$, $\mu = 0$ and $\Sigma = [-S, S]$, \cf \cite{RotatingString}. Higher-dimensional examples can be straightforwardly obtained by generalizing the Nambu-Goto action to higher dimensional objects, i.e., branes.

Interestingly, actions of the above type were also considered in the context of the AdS/CFT correspondence. In treatments of the gauge/gravity duality, one supplements the bulk Einstein-Hilbert action with cosmological constant with a counterterm boundary Einstein-Hilbert action with cosmological constant \cite{BalasubramanianKraus99}. The role of the length scale $c$ is there played by the radius of curvature of the AdS space-time. Note, however, that the counterterm action is fixed by the requirement of obtaining a finite stress-energy tensor by variations of the metric on the boundary, whereas in \eqref{eq:action} $c$ can be chosen arbitrarily. Also note that the analogue of $c$ is negative for the gravitational counterterms. In holographic renormalization \cite{Skenderis02}, also boundary counterterms for scalar fields are introduced, of the same form as in \eqref{eq:action}. Still, $c$ is negative, but it is no longer simply given by the radius of curvature. As we will see, negative $c$ leads to severe difficulties already in the treatment of the classical system.

Another instance where boundary terms of the above form arise naturally is in models with a non-renormalizable bulk action. For a boundary of co-dimension $k$, the power-counting degree of divergence of terms localized at the boundary is reduced by $k$. This was discussed in \cite{Symanzik81} in the context of a scalar field with Dirichlet boundary conditions. Hence, in renormalizable theories, where the field strength counterterm has at most a logarithmic divergence, no kinetic counterterm is necessary at the boundary. This, however, changes in non-renormalizable theories, where the occurrence of a kinetic term at the boundary is thus generic.

Variation of the action \eqref{eq:action} yields, upon integration by parts, the equations of motion
\begin{align}
\label{eq:eom_bk}
 - \Box_g \phi + \mu^2 \phi & = 0 & \text{ in } M, \\
\label{eq:eom_bd}
 - \Box_h \phi + \mu^2 \phi & = c^{-1} \del_\perp \phi & \text{ in } \del M.
\end{align}
Here $\del_\perp$ denotes the inward pointing normal derivative. Using \eqref{eq:eom_bk}, one may write \eqref{eq:eom_bd} alternatively as
\begin{align}
\label{eq:bdyCondition}
 \del_\perp^2 \phi & = c^{-1} \del_\perp \phi & \text{ in } \del M.
\end{align}
These equations may be read as an equation of motion for the bulk \eqref{eq:eom_bk}, supplemented by boundary conditions \eqref{eq:eom_bd} or equivalently \eqref{eq:bdyCondition}. Boundary conditions of this type are known as generalized Wentzell or Feller-Wentzell type boundary conditions \cite{Ueno73,FaviniGoldstein2Romanelli02}.\footnote{I am grateful to Konstantin Pankrashkin for pointing out to me the mathematical literature on the subject. Note, however, that in the context of the wave equation, the generalized Wentzell boundary conditions considered in the mathematical literature are usually slightly different, with $\Box_h$ on the \lhs of \eqref{eq:eom_bd} replaced by $\del_0^2$ \cite{GalGoldstein203}, i.e., the terminology used here does only coincide with the one used in some of the literature in the case $d = 1$. An example where generalized Wentzell boundary conditions are formulated so general that the case \eqref{eq:eom_bd} is also covered for $d > 1$ is \cite{CocliteFaviniGoldstein2Romanelli14}. The case studied here, but with $\mu = 0$ and additional damping or source terms, was studied in \cite{Vitillaro13, Vitillaro15}, where these boundary conditions were called ``kinetic'' or ``dynamical''.}

It turns out that it is physically more appropriate to consider \eqref{eq:eom_bk} and \eqref{eq:eom_bd} as two wave equations, which are coupled by the solution of the bulk equation \eqref{eq:eom_bk} providing a source on the \rhs of the boundary equation \eqref{eq:eom_bd} and the solution of the boundary equation \eqref{eq:eom_bd} providing a Dirichlet type boundary condition for the bulk equation \eqref{eq:eom_bk}. This is reminiscent of the fact that in the AdS/CFT correspondence the boundary values of the bulk fields act as sources for the boundary field.
Mathematically, this interpretation is implemented by considering Cauchy data in $(L^2(\Sigma) \oplus L^2(\del \Sigma))^2$, or related Sobolev spaces. The fact that the Hilbert space $L^2(\Sigma) \oplus L^2(\del \Sigma)$ is appropriate for Wentzell boundary conditions was already noted in \cite{Feller57,Ueno73,FaviniGoldstein2Romanelli02,GalGoldstein203,Vitillaro13}. Physically, this means that the boundary can carry energy.
We will prove well-posedness of the Cauchy problem and causal propagation. This may be interesting in itself, as it contradicts the folklore wisdom that Robin boundary conditions are the most general sensible linear boundary conditions for the wave equation.

As a by-product of the discussion of the wave equation \eqref{eq:eom_bk}, \eqref{eq:eom_bd}, we obtain a (generalized) orthonormal basis of solutions. This can be used for a canonical quantization of the system.
The unusual space of Cauchy data then has interesting physical implications. In particular, it turns out that it is possible to restrict the  field to the boundary in a natural way, resulting on a dimensionally reduced field, which, interestingly, has the short-distance behavior of a free scalar field theory on $\del M$. Furthermore, bulk observables can be holographically mapped to boundary observables. This is possible as the boundary field is a generalized free field \cite{Greenberg}.
 Heuristically, the transversal degree of freedom that is lost by the restriction to the boundary is traded for the possibility to excite these higher modes. 
 
The classical aspects of the field equations \eqref{eq:eom_bk}, \eqref{eq:eom_bd} are discussed in the next section.
In Section~\ref{sec:Quantization}, the quantization of the system is performed and the holographic aspects are discussed. In Section~\ref{sec:Comments}, we comment on the relation to other boundary conditions and holography on AdS.

\subsection*{Notation and conventions}

We use signature $(-, +, \dots, +)$. The coordinates on $\Sigma = \R^{d-1} \times \R_+$ or $\R^{d-1} \times [-S, S]$ are usually denoted $(x, z)$. For $x \in \R^{1, d-1}$, the spatial part is denoted by $\underline{x}$. In the case $\Sigma = \R^{d-1} \times [-S, S]$, $\del_\pm \Sigma$ denotes the component at $\pm S$, and analogously for $\del_\pm M$. Fourier transformation is defined as
\[
 \hat f(k) = (2 \pi)^{- \frac{d}{2}} \int f(x) e^{- i k x} \ud^d x.
\]
The sign of the Laplacian is defined as $\Delta = \del_i \del_i$. For a set $\Sigma \subset M$, we denote by $D^+(\Sigma)$ its future domain of dependence, i.e., the set of all points $p$ such that all past inextendable causal piecewise $C^1$ curves through $p$ intersect $\Sigma$. These curves may also be part of the boundary. As usual, $H^s(\R^d)$ and $\sS(\R^d)$ denote Sobolev and Schwartz spaces, the latter being the space of rapidly decreasing smooth functions.

\section{The wave equation}
\label{sec:WaveEquation}

Let us study the wave equation \eqref{eq:eom_bk}, \eqref{eq:eom_bd}. It is straightforward to check that the symplectic form 
\begin{equation}
\label{eq:SymplecticForm}
 \sigma((\phi, \dot \phi), (\psi, \dot \psi)) = \int_{\Sigma} \phi \dot \psi - \dot \phi \psi + c \int_{\del \Sigma} \phi \dot \psi - \dot \phi \psi
\end{equation}
is conserved for solutions $\phi$, $\psi$. It is thus natural to introduce the scalar product
\[
 \skal{\phi}{\psi} = \int_{\Sigma} \bar \phi \psi + c \int_{\del \Sigma} \bar \phi \psi
\]
on functions on $\Sigma$. Completion in the corresponding norm yields the Hilbert space $H = L^2(\Sigma, \varrho + c \delta_{\del \Sigma})$, where $\varrho$ is the Lebesgue measure, and $\delta_{\del \Sigma}$ the Dirac measure on the boundary. Under completion of the space of continuous functions, bounded functions that are localized at the boundary $\del \Sigma$ are not equivalent to zero. Hence, $H$ is canonically isomorphic to $L^2(\Sigma) \oplus L^2(\del \Sigma)$. An element of $H$ is typically written as $\Phi = (\phi, \phi|)$ (but note that $\phi|$ need not coincide with the boundary value $\phi|_{\del \Sigma}$ of $\phi$, even if it is well-defined), with the scalar product
\begin{equation}
\label{eq:ScalarProduct}
 \skal{\Phi}{\Psi} = \skal{\phi}{\psi}_{L^2(\Sigma)} + c \skal{\phi|}{\psi|}_{L^2(\del \Sigma)}.
\end{equation}
On $H$, the equation of motion \eqref{eq:eom_bk}, \eqref{eq:eom_bd} can be written as
\[
 - \del_t^2 \Phi = \Delta \Phi,
\]
with
\[
 \Delta = \begin{pmatrix} - \Delta_\Sigma + \mu^2 & 0 \\ - c^{-1} \del_\perp \cdot | & - \Delta_{\del \Sigma} + \mu^2 \end{pmatrix},
\]
where we encode the boundary condition in the domain
\beq
\label{eq:DeltaDomain}
 D = \left\{ (\phi, \phi|) \in H \mid \phi \in H^2(\Sigma), \phi| \in H^2(\del \Sigma), \phi|_{\del \Sigma} = \phi| \right\}.
\eeq
Note that the Sobolev space $H^s(\Sigma)$ is given by \cite[Section~4.5]{TaylorPDEsI}
\[
 H^s(\Sigma) = \frac{H^s(\R^d)}{ \left\{ u \in H^s(\R^d) \mid \supp u \subset \overline{ \R^d \setminus \Sigma } \right\} }.
\]
Also note that for $s > \frac{1}{2}$, the restriction map $H^s(\Sigma) \to H^{s-\frac{1}{2}}(\del \Sigma)$ is well-defined and continuous \cite[Prop.~4.4.5]{TaylorPDEsI}, so that the boundary condition makes sense.

In the following, we consider the case of the half-space $\Sigma = \R_+^{d}$.
 The case $\Sigma = \R^{d-1} \times [-S, S]$ is discussed at the end of the section.
\begin{proposition}
\label{prop:SelfAdjoint}
With $\Sigma = \R^d_+$ and on the domain $D$, $\Delta$ is self-adjoint with spectrum contained in $[\mu^2, \infty)$. A normalized\footnote{Normalization is understood in the distributional sense, i.e., $\skal{\Phi_p}{\Phi_{p'}} = \delta(p-p')$.} complete system of generalized eigenfunctions of $\Delta$ is given by
\beq
\label{eq:Phi_k_q}
 \Phi_{k, q} = \left( (2 \pi)^{d-1} \tfrac{\pi}{2} \left( c^2 q^2 + 1 \right) \right)^{-\frac{1}{2}}  \left( e^{i k x} (\cos q z - c q \sin q z), e^{i k x} \right),
\eeq
with $k \in \R^{d-1}$, $q \in \R_+$ and the eigenvalue
\beq
\label{eq:omega_k_q}
 \omega^2_{k, q} = k^2 + q^2 + \mu^2.
\eeq
\end{proposition}
\begin{proof}
We compute (for simplicity, we here set $\mu = 0$)
\begin{align*}
 \skal{\Phi}{\Delta \Psi} & = - \int_\Sigma \bar \phi \Delta_\Sigma \psi - \int_{\del \Sigma} \bar \phi| \del_\perp \psi|_{\del \Sigma} - c \bar \phi| \Delta_{\del \Sigma} \psi| \\
 & = \int_\Sigma \del_i \bar \phi \del_i \psi + \int_{\del \Sigma} \bar \phi|_{\del \Sigma} \del_\perp \psi|_{\del \Sigma} - \bar \phi| \del_\perp \psi|_{\del \Sigma} - c \Delta_{\del \Sigma} \bar \phi| \psi|  \\
 & = - \int_\Sigma \Delta_\Sigma \bar \phi \psi - \int_{\del \Sigma} \del_\perp \bar \phi|_{\del \Sigma} \psi|_{\del \Sigma} - \bar \phi|_{\del \Sigma} \del_\perp \psi|_{\del \Sigma} + \bar \phi| \del_\perp \psi|_{\del \Sigma} + c \Delta_{\del \Sigma} \bar \phi| \psi|,
\end{align*}
which for $\psi| = \psi|_{\del \Sigma}$ equals $\skal{\Delta \Phi}{\Psi}$ iff $\phi| = \phi|_{\del \Sigma}$. In particular, it follows that the domain of $\Delta^*$ is contained in $\left\{ (\phi, \phi|) \mid \phi|_{\del \Sigma} = \phi| \right\}$. But for $\Delta^*$ to be well-defined on $(\phi, \phi|)$, we also need $\phi \in H^2(\Sigma)$, $\phi| \in H^2(\del \Sigma)$.
For the claim on the spectrum, we compute, for $\Psi \in D$,
\[
 \skal{\Psi}{\Delta \Psi} = \int_\Sigma \del_i \bar \psi \del_i \psi + c \int_{\del \Sigma} \del_a \bar \psi| \del_a \psi| \geq 0,
\]
where the index $a$ runs over the coordinates on $\del \Sigma$.
A separation ansatz for the generalized eigenfunctions of $- \Delta_\Sigma$ is
\begin{align}
\label{eq:phi_1bdy}
 \phi(x, z) & = e^{i k x} \left( A \cos q z + B \sin q z \right),
\end{align}
where we may assume that $k \in \R^{d-1}$ is real. Obviously,
\[
 \begin{pmatrix} - \Delta_\Sigma + \mu^2 & 0 \\ - c^{-1} \del_\perp \cdot |_{\del \Sigma} & - \Delta_{\del \Sigma} + \mu^2 \end{pmatrix} \begin{pmatrix} \phi \\ \phi|_{\del \Sigma} \end{pmatrix} = \omega^2 \begin{pmatrix} \phi \\ \phi|_{\del \Sigma} \end{pmatrix}
\]
with $\omega^2 = k^2 + q^2 + \mu^2$ implies
\[
 B = - c q A.
\]
From this, the generalized basis \eqref{eq:Phi_k_q} follows by normalization.
\end{proof}

In order to discuss the regularity of the solutions, it is advantageous to restrict to $\mu > 0$ and to define the admissible spaces of Cauchy data as
\[
 \sK_r = \sD_{r+1} \oplus \sD_r,
\]
where for $r \in \N_0$,
\[
 \sD_r = \dom \Delta^\frac{r}{2}
\]
is a Hilbert space with inner product
\beq
\label{eq:skal_D_r}
 \skal{ \cdot}{\cdot}_{\sD_r} = \skal{\Delta^{\frac{r}{2}} \cdot}{\Delta^{\frac{r}{2}} \cdot}.
\eeq
For $r \in \Z, r< 0$, one defines $\sD_r$ as the completion of $H$ \wrt the inner product \eqref{eq:skal_D_r}. In particular, $\sD_r^* = \sD_{-r}$. Obviously,
\[ 
 \Delta \sD_r = \sD_{r-2}.
\]
The equation of motion can on $\sK_r$ now be written as
\[
 i \del_t \begin{pmatrix} \Phi_0 \\ \Phi_1 \end{pmatrix} = \begin{pmatrix} 0 & 1 \\ \Delta & 0 \end{pmatrix} \begin{pmatrix} \Phi_0 \\ \Phi_1 \end{pmatrix} = A \begin{pmatrix} \Phi_0 \\ \Phi_1 \end{pmatrix}.
\]
The operator $A$ is self-adjoint on $\sK_r$ with domain
\[
 \dom(A) = \dom \Delta^{\frac{r+2}{2}} \oplus \dom \Delta^{\frac{r+1}{2}}.
\]
In particular, the equation of motion is solved by
\[
 \begin{pmatrix} \Phi_0(t) \\ \Phi_1(t) \end{pmatrix} = e^{- i A t} \begin{pmatrix} \Phi_0(0) \\ \Phi_1(0) \end{pmatrix}.
\]
The time-evolution leaves the domain invariant.

It is instructive to determine the domain of $\Delta^2$. Obviously,
\[
 \dom(\Delta^2) = \left\{ \Phi \in \dom(\Delta) \mid \Delta \Phi \in \dom(\Delta) \right\}.
\]
In order to fulfill $\Delta \Phi \in \dom(\Delta)$, we certainly have to require that $\phi \in H^4(\Sigma)$, $\phi| \in H^4(\del \Sigma)$. However, we also have to ensure that for $\Delta \Phi$ the boundary value of the bulk field coincides with the boundary field. This yields
\begin{multline*}
 \dom(\Delta^2) \\ = \left\{ (\phi, \phi|) \mid \phi \in H^4(\Sigma), \phi| \in H^4(\del \Sigma), \phi|_{\del \Sigma} = \phi|, \del_\perp^2 \phi|_{\del \Sigma} = c^{-1} \del_\perp \phi|_{\del \Sigma} \right\},
\end{multline*}
i.e., the bulk field has to fulfill the boundary condition \eqref{eq:bdyCondition}. Analogously, one finds
\begin{multline*}
 \dom(\Delta^k) = \left\{ (\phi, \phi|) \mid \phi \in H^{2k}(\Sigma), \phi| \in H^{2k}(\del \Sigma), \phi|_{\del \Sigma} = \phi|, \right. \\
 \left. \del_\perp^{2j-2} \phi|_{\del \Sigma} = c^{-1} \del_\perp^{2j-3} \phi|_{\del \Sigma} \ \forall j \leq k \right\}.
\end{multline*}

Now consider Cauchy data $\Psi = (\Phi_0, \Phi_1) \in \sK_\infty$, with
\[
 \sK_\infty = \cap_{r} \sK_r.
\]
For time derivatives of the corresponding solutions $\Psi(t) = e^{- i A t} \Psi$, we compute
\[
 \norm{\del_t^m  \Psi(t)}^2_{\sK_s} = \norm{ A^m \Psi(t) }^2_{\sK_s} = \norm{\Psi(t)}^2_{\sK_{s+m}} = \norm{\Psi(0)}^2_{\sK_{s+m}},
\]
where the last equality follows from the unitarity of the time evolution. We can thus bound arbitrary derivatives on future Cauchy surfaces by the initial data. It is instructive to compute
\begin{align*}
 \norm{\Phi}^2_{\sD_1} & = \skal{\Phi}{\Delta \Phi} \\
 & = \skal{\phi}{( - \Delta_\Sigma + \mu^2 ) \phi}_{L^2(\Sigma)} + c \skal{\phi|}{( - \Delta_{\del \Sigma} + \mu^2 ) \phi|}_{L^2(\del \Sigma)} - \skal{\phi|}{\del_\perp \phi}_{L^2(\del \Sigma)} \\
 & = \norm{ \nabla_\Sigma \phi}^2_{L^2(\Sigma)} + \mu^2 \norm{\phi}^2_{L^2(\Sigma)} + c \norm{ \nabla_{\del \Sigma} \phi|}^2_{L^2(\del \Sigma)} + c \mu^2 \norm{\phi|}^2_{L^2(\del \Sigma)},
\end{align*}
where we assumed that $\phi|_{\del \Sigma} = \phi|$. Using $\norm{\Phi}^2_{\sD_k} = \norm{\Delta \Phi}^2_{\sD_{k-2}}$, we obtain, for $k$ odd and $\Phi \in \dom(\Delta^{\frac{k+1}{2}})$,
\begin{multline}
\label{eq:norm_D_k_odd} 
\norm{\Phi}^2_{\sD_k} = \norm{ \nabla_\Sigma (- \Delta_\Sigma + \mu^2)^{\frac{k-1}{2}} \phi}^2_{L^2(\Sigma)} + \mu^2 \norm{ (- \Delta_\Sigma + \mu^2)^{\frac{k-1}{2}} \phi}^2_{L^2(\Sigma)} \\
+ c \norm{ \nabla_{\del \Sigma} (- \Delta_\Sigma + \mu^2)^{\frac{k-1}{2}} \phi|_{\del \Sigma}}^2_{L^2(\del \Sigma)} + c \mu^2 \norm{ (- \Delta_\Sigma + \mu^2)^{\frac{k-1}{2}} \phi|_{\del \Sigma}}^2_{L^2(\del \Sigma)}.
\end{multline}
For even $k$ and $\Phi \in \dom(\Delta^{\frac{k}{2}+1})$, we analogously obtain
\begin{equation}
\label{eq:norm_D_k_even}
\norm{\Phi}^2_{\sD_k} =
 \norm{(-\Delta_\Sigma + \mu^2)^{\frac{k}{2}} \phi}^2_{L^2(\Sigma)} + c \norm{(-\Delta_\Sigma + \mu^2)^{\frac{k}{2}} \phi|_{\del \Sigma}}^2_{L^2(\del \Sigma)}.
\end{equation}
Denoting by $H^\infty( \cdot) = \cap_r H^r( \cdot )$ the intersection of Sobolev spaces, we have thus shown:

\begin{proposition}
For smooth Cauchy data
\[
 (\phi_0, \phi_1) \in H^\infty(\R^d_+) \times H^\infty(\R^d_+)
\]
such that
\begin{equation}
\label{eq:BoundaryConditionInitialData}
 \del_\perp^{2 k+2} \phi_i|_{\del \Sigma} = c^{-1} \del_\perp^{2 k + 1} \phi_i|_{\del \Sigma}, \qquad \forall k \in \N,
\end{equation}
for $i = 0, 1$, there is a unique smooth solution $\phi(t)$ to the wave equation \eqref{eq:eom_bk}, \eqref{eq:eom_bd} with $\mu > 0$. The properties of the Cauchy data are conserved under time evolution. Furthermore, denoting $\Phi(t) = (\phi(t), \phi(t)|_{\del \Sigma})$, we have
\begin{equation}
\label{eq:CauchyBoundGlobal}
 \norm{\del_t^m \Phi(t)}^2_{\sD_{k+1}} + \norm{\del_t^{m+1} \Phi(t)}^2_{\sD_k} = \norm{\Phi_0}^2_{\sD_{k+m+1}} + \norm{\Phi_1}^2_{\sD_{k+m}}.
\end{equation}
\end{proposition}

As usual, local energy estimates are very useful to prove causal propagation. As can be expected from the action \eqref{eq:action}, the boundary should be taken into account with its own weight:
\begin{proposition}
\label{prop:EnergyEstimate}
Let $\Sigma_0$, $\Sigma_1$ be two equal time surfaces, with $\Sigma_1$ in the future of $\Sigma_0$. Let $S_0 \subset \Sigma_0$ and $S_1 = D^+(S_0) \cap \Sigma_1$, \cf Figure~\ref{fig:CausalDomain}. Then for a solution $\phi$ to \eqref{eq:eom_bk}, \eqref{eq:eom_bd}, it holds
\begin{multline}
\label{eq:EnergyEstimate}
 \int_{S_1} (\del_0 \phi)^2 + g^{i j} \del_i \phi \del_j \phi + \mu^2 \phi^2 + c \int_{S_1 \cap \del M} (\del_0 \phi)^2 + h^{i j} \del_i \phi \del_j \phi + \mu^2 \phi^2 \\
 \leq \int_{S_0} (\del_0 \phi)^2 + g^{i j} \del_i \phi \del_j \phi + \mu^2 \phi^2 + c \int_{S_0 \cap \del M} (\del_0 \phi)^2 + h^{i j} \del_i \phi \del_j \phi + \mu^2 \phi^2.
\end{multline}
\end{proposition}

\begin{figure}
\begin{center}
\includegraphics[width=0.8\textwidth]{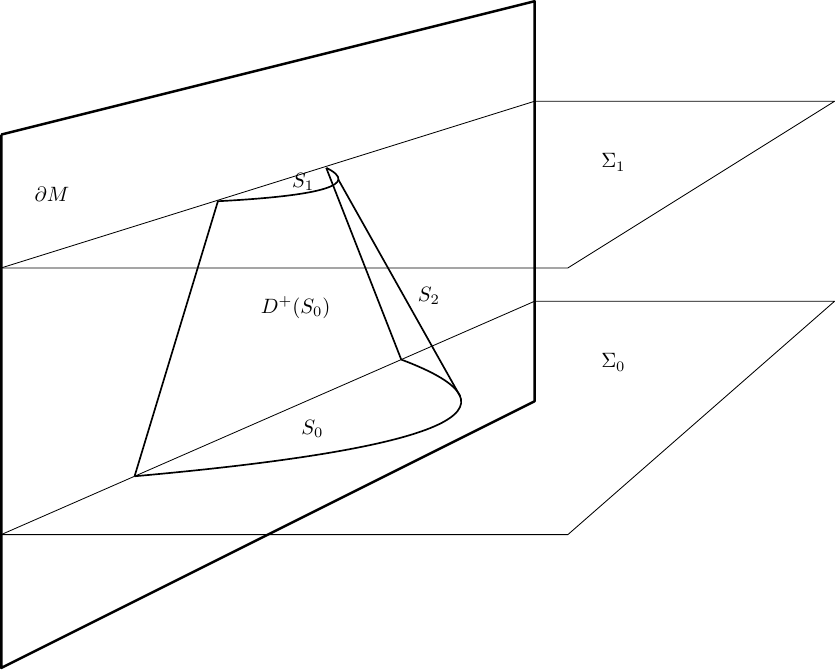}
\caption{Illustration of the geometric setup in Proposition~\ref{prop:EnergyEstimate}.}
\label{fig:CausalDomain}
\end{center}
\end{figure}

\begin{proof}
We consider the bulk and boundary stress-energy tensors
\begin{align*}
 T_{\mu \nu} & = \del_\mu \phi \del_\nu \phi - \tfrac{1}{2} g_{\mu \nu} \left( \del_\lambda \phi \del^\lambda \phi + \mu^2 \phi^2 \right), \\
 T|_{a b} & = c \left[ \del_a \phi \del_b \phi - \tfrac{1}{2} h_{a b} \left( \del_c \phi \del^c \phi + \mu^2 \phi^2 \right) \right],
\end{align*}
where roman indices refer to the coordinates on $\del M$. Obviously, $T_{\mu \nu}$ is conserved on-shell. For the boundary stress-energy tensor one finds
\beq
\label{eq:divT_bdy}
 \del^a T|_{a b} = - T_{\perp b}.
\eeq
Furthermore, both $T_{\mu \nu}$ and $T|_{a b}$ fulfill the dominant energy condition, i.e., for future pointing time-like $\xi^\mu$, $\eta^a$, we have
\begin{align*}
 T_{\mu \nu} \xi^\mu \xi^\nu & \geq 0, & T^\mu_\nu \xi^\nu \text{ time-like or null}, \\ 
 T|_{a b} \eta^a \eta^b & \geq 0, & T^a_b \eta^b \text{ time-like or null}.
\end{align*}
Now we choose $\xi = e^0$ and integrate $\del^\mu T_{\mu \nu} \xi^\nu$ over $D = D^+(S_0) \cap J^-(\Sigma_1)$, obtaining
\[
 \int_{S_0} T_{00} = \int_{S_1} T_{00} + \int_{S_2} \ell^\mu T_{\mu 0} - \int_{\del D} T_{\perp 0},
\]
where $\ell$ is the future directed normal to $S_2$, \cf Figure~\ref{fig:CausalDomain}. We also choose $\eta = e^0$ and integrate $\del^a T|_{a b} \eta^b$ over $\del D$, obtaining
\[
 \int_{S_0 \cap \del M} T|_{00} = \int_{S_1 \cap \del M} T|_{00} + \int_{S_2 \cap \del M} p^a T|_{a 0} + \int_{\del D} T_{\perp 0},
\]
where we used \eqref{eq:divT_bdy}, and $p$ is the future directed normal to $S_2 \cap \del M$. Summing the previous two equations and using the dominant energy condition, one finds \eqref{eq:EnergyEstimate}.
\end{proof}

As usual, such energy estimates establish the causal propagation of smooth solutions. Because of the causal propagation, one may refine \eqref{eq:CauchyBoundGlobal} by considering the norms on the \lhs \wrt the Sobolev spaces on $S_1$ and $S_1 \cap \del M$. Given arbitrary smooth initial data that fulfills \eqref{eq:BoundaryConditionInitialData}, one can construct solutions up to time $T$ as follows: Consider a dense enough covering of the Cauchy surface with open regions $U_i$ of diameter $3T$. For each $i$, cut off the initial data outside of $U_i$ and solve the wave equation. Then glue the solutions together. Hence, we obtain:

\begin{proposition}
The wave equation \eqref{eq:eom_bk}, \eqref{eq:eom_bd} for $\Sigma = \R^d_+$ and $\mu > 0$ is well-posed for smooth initial data $(\phi_0, \phi_1) \in C^\infty(\R^d_+) \times C^\infty(\R^d_+)$ fulfilling \eqref{eq:BoundaryConditionInitialData}, i.e., there exists a unique smooth solution $\phi$ which depends continuously and causally on the initial data  in the sense that, for $\Phi(t) = (\phi(t), \phi(t)|_{\del \Sigma})$,
\begin{equation}
\label{eq:CauchyBoundCausal}
 \norm{\del_t^m \Phi(t)}^2_{\sD_{k+1}(S_1)} + \norm{\del_t^{m+1} \Phi(t)}^2_{\sD_k(S_1)} \leq \norm{\Phi_0}^2_{\sD_{k+m+1}(S_0)} + \norm{\Phi_1}^2_{\sD_{k+m}(S_0)}.
\end{equation}
Here we used the same convention as in Proposition~\ref{prop:EnergyEstimate}, and $\norm{\Phi}^2_{\sD_k(S)}$ stands for the restriction to $S$ and $S \cap \del \Sigma$ of the integrals in \eqref{eq:norm_D_k_odd}, \eqref{eq:norm_D_k_even}.
\end{proposition}

As usual, \cf \cite{WaldGR} for example, the estimates on the \lhs of \eqref{eq:CauchyBoundCausal} may be converted to supremum estimates using the Sobolev embedding theorems. The causal propagation can also be established for distributions, using the duality of $\sD_r$ and $\sD_{-r}$, analogously to the Dirichlet case discussed in \cite[Sect.~6.1]{TaylorPDEsI}.

The global energy estimate \eqref{eq:CauchyBoundGlobal} for $\mu = 0$ and $m = k = 0$ was already derived in \cite{Vitillaro15}. However, to the best of my knowledge, local energy estimates and thus causal propagation have not yet been proven in the literature.

\begin{remark}
With small adjustments, one can also treat the massless case $\mu = 0$. The technical difficulty is that $\Delta$ is not strictly positive so that $\sD_r = \dom \Delta^{\frac{r}{2}}$ is not complete \wrt the scalar product \eqref{eq:skal_D_r}. One may for example make use of the causal propagation to restrict the discussion about propagation in a bounded region to the case of a supplementary Dirichlet boundary at some $z = L$. This yields a strictly positive $\Delta$, so that the above techniques can be applied. Also the case with different masses for the bulk and the boundary may be treated. For $\mu_\bk \geq \mu_\bd$, the adjustments are minor. However, for $\mu_\bk < \mu_\bd$, there will be a bound state, but $\Delta$ is still strictly positive for $\mu_\bk^2 > 0$.
\end{remark}

\begin{remark}
For the case $c < 0$, there seem to be severe difficulties. First of all, $H$ is then no longer a Hilbert, but a Krein space. This may be possible to deal with, one would have to show that $\Delta$ is definitizable and regular at infinity. But even then the energy estimates will no longer work, so it will be difficult, if not impossible, to establish continuous dependence on the initial data and causality.
\end{remark}

As can be seen by the form of the energy estimates, the boundary may carry some energy. It is instructive to consider a concrete example. Consider $\mu = 0$ and a singularity
\[
 \phi = \delta(t+z)
\] 
infalling to the boundary for $t < 0$. The full solution to the equations \eqref{eq:eom_bk}, \eqref{eq:eom_bd} is then (see below)
\begin{align}
\label{eq:ExplicitSolution}
 \phi & = \delta(t+z) - \delta(t-z) + 2 c^{-1} e^{-\frac{t-z}{c}} \theta(t-z),  \\
 \phi| & = 2 c^{-1} e^{-\frac{t}{c}} \theta(t), \nonumber
\end{align}
with $\theta$ the Heaviside function. We see that the boundary absorbs some energy and radiates it off on the time-scale $c$. It is also obvious from this example why negative $c$ seems physically unacceptable.

The solution \eqref{eq:ExplicitSolution} has to be understood in the weak sense, i.e., for any test function $\vp \in \D(\R^{d+1})$ with $\Box_h \vp|_{z = 0} = - c^{-1} \del_z \vp|_{z=0}$, we have
\[
 \int_{\R^{d+1}_+} \phi \Box \vp = 0.
\]
To check this, it is convenient to introduce light cone coordinates
\begin{align*}
 u & = t + z, &
 v & = t - z.
\end{align*}
For convenience, we may assume $d=1$. We then compute
\begin{align*}
 \int_{\R^{2}_+} \phi \Box \vp & = - 2 \int_{v < u} \left( \delta(u) - \delta(v) + 2 c^{-1} e^{-\frac{v}{c}} \theta(v) \right) \left( \del_u \del_v \vp \right) \ud u \ud v \\
 & = - 2 \del_u \vp(0,0) - 2 \del_v \vp(0,0) + 4 c^{-1} \int_0^\infty e^{-\frac{v}{c}} \del_v \vp|_{u = v} \ud v \\
 & = - 2 \del_t \vp(0,0) + 2 c^{-1} \int_0^\infty e^{-\frac{t}{c}} \left( \del_t - \del_z \right) \vp(t, 0) \ud t \\
 & = - 2 \del_t \vp(0,0) + 2 c^{-1} \int_0^\infty e^{-\frac{t}{c}} \left( \del_t - c \del_t^2 \right) \vp(t, 0) \ud t,
\end{align*}
where in the last step we used the boundary condition for $\vp$. One easily checks by integration by parts that this vanishes.

\begin{remark}
\label{rem:BoundaryBehaviorGeneral}
The fact that the restriction of the normalized modes with transversal momentum $q$ falls off like $q^{-1}$, \cf \eqref{eq:Phi_k_q}, will be important in the following. This property is generic and does not depend on the specific form of the boundary, the equality of the two masses in the bulk and the boundary or the presence of curvature and curvature couplings. Choosing normal coordinates, we may write the boundary condition as
\[
 \del_\perp^2 \phi = c^{-1} \del_\perp \phi + R \phi,
\]
with $R$ some operator that does not act in the transversal direction. For large transversal momentum the mode is well approximated by
\[
 \phi_{q,k}(x) \simeq C f_k(x) \sin(q z + \phi_q)
\]
with $f_k$ normalized modes on the boundary and a normalization constant $C$ which is essentially independent of $q$. We conclude that for large transversal momentum, where the term $R \phi$ may be neglected, and the boundary at $z = 0$, we have
\[
 \tan (\phi_q) \simeq (c q)^{-1}.
\]
For large $q$, this means
\[
 C \sin(\phi_q) \simeq C (c q)^{-1},
\]
so that the normalization constant of the modes on the boundary indeed falls off like $\frac{1}{q}$.
\end{remark}

We close this section by briefly discussing the case of the strip $\Sigma = \R^{d-1} \times [-S, S]$.
The statement analogous to Proposition~\ref{prop:SelfAdjoint} is now:
\begin{proposition}
\label{prop:SelfAdjointStrip}
With $\Sigma = \R^{d-1} \times [-S, S]$ and on the domain $D$, $\Delta$ is self-adjoint with spectrum contained in $[\mu^2, \infty)$. A normalized complete system of generalized eigenfunctions of $\Delta$ is given by
\begin{align}
\label{eq:Phi_k_m}
 \phi_{k, m} & = c_m (2 \pi)^{- \frac{d-1}{2}} S^{-\frac{1}{2}} e^{i k x} \begin{cases} \cos q_m z & m \text{ even } \\ \sin q_m z & m \text{ odd }\end{cases} \\
 \Phi_{k, m} & = (\phi_{k, m}, \phi_{k, m}|_{\del \Sigma}), \nonumber
\end{align}
with $k \in \R^{d-1}$, $m \in \N$ and the eigenvalue
\beq
\label{eq:omega_k_m}
 \omega^2_{k, m} = k^2 + q_m^2 + \mu^2.
\eeq
Here $\{ q_m \}$ is an increasing sequence of non-negative real numbers with $q_0 = 0$ and
\beq
\label{eq:q_Interval}
 q_{2 p} \in ((p-\tfrac{1}{2})\tfrac{\pi}{S},p \tfrac{\pi}{S}), \qquad q_{2 p-1} \in ((p-1)\tfrac{\pi}{S},(p-\tfrac{1}{2}) \tfrac{\pi}{S})
\eeq
for all $p \geq 1$. For large enough $m$, the $q_m$ are bounded by
\beq
\label{eq:q_bound}
 \frac{\pi}{2 S} (m-1) + (1-\delta) \frac{2 c^{-1}}{\pi (m-1)} \leq q_m \leq \frac{\pi}{2 S} (m-1) + \frac{2 c^{-1}}{\pi (m-1)}  
\eeq
for any $\delta > 0$. For large $m$, the normalization constants behave as $c_m = 1 + \order(m^{-2})$. The restriction to the boundary is given by
\beq
\label{eq:phi_k_m_restricted}
 \phi_{k, m}|_{\del_{\pm} \Sigma}(x) = (\pm)^m (2 \pi)^{-\frac{d-1}{2}} d_m e^{i k x}
\eeq
where the $d_m$ are real, non-zero and fulfill
\[
 (1-\delta) \frac{2 c^{-1} \sqrt{S}}{\pi (m-1)} \leq \betrag{d_m} \leq (1+\delta) \frac{2 c^{-1} \sqrt{S}}{\pi (m-1)}
\]
for any $\delta >0$ and large enough $m$.
\end{proposition}
\begin{proof}
The statement on self-adjointness and the spectrum of $\Delta$ follows as in the proof of Proposition~\ref{prop:SelfAdjoint}.
Due to the symmetry of the problem, the separation ansatz \eqref{eq:Phi_k_m} is general enough. For the even/odd modes, the boundary condition then implies
\begin{align}
\label{eq:EV_even}
 c^{-1} \tan q S & = - q, \\
\label{eq:EV_odd}
 q \tan q S & = c^{-1},
\end{align}
which only have real solutions $q$ due to $\Im \tan(x+i y) \gtrless 0$ for $y \gtrless 0$.
The statement \eqref{eq:q_Interval} for odd $m$ follows from the monotonicity of $\tan$ on the interval $(\pi (p-\frac{3}{2}), \pi (p-\frac{1}{2}))$. Hence, for odd $m$, we must have $q_m = \frac{\pi}{2 S}((m-1) + \eps_m)$ with some $\eps_m > 0$. Clearly, for $\frac{\pi}{2 S} (m-1) \gg c^{-1}$, we must have $\eps_m \ll 1$. Using that
\[
 \eps \leq \tan (\pi N + \eps) \leq (1+\delta) \eps
\]
for $N \in \N$, $\delta>0$ and $\eps$ small enough, we have
\[
 \frac{\pi^2}{4 S} (m-1) \eps_m \leq q_m \tan S q_m \leq (1+\delta) \frac{\pi^2}{4 S} (m-1) \eps_m 
\]
for $m$ large enough. With \eqref{eq:EV_odd}, this can be used to bound $\eps_m$ to show \eqref{eq:q_bound}. For even $m$ one argues analogously. For the statement on the normalization, we compute, for odd $m$,
\[
 \skal{\Phi_{k,m}}{\Phi_{k',m}} = \betrag{c_m}^2 S^{-1} \delta(k-k') \left[ \left( S - \frac{\sin(2 q_m S)}{2 q_m} \right) + 2 c \sin^2(q_m S) \right]
\]
By \eqref{eq:q_bound}, the expression in square brackets on the \rhs is $S + \order(m^{-2})$. The bounds on $d_m$ then follow again from \eqref{eq:q_bound} and the behavior of $\sin x$ near $x= N \pi$. For even $m$, one argues analogously. That the $d_m$ are non-zero follows from \eqref{eq:q_Interval}.
\end{proof}

With this result, one may continue as for the half-space to establish well-posedness of the Cauchy problem and causal propagation.


\section{Quantization}
\label{sec:Quantization}

We begin to study the quantization for the case $\Sigma = \R^{d-1} \times [-S, S]$, where the holographic mapping has nicer properties than for the half-space $\R^d_+$.

Given the orthonormal basis $\{ \Phi_{k,m} \}_{k \in \R^{d-1}, m \in \N}$, \cf Proposition~\ref{prop:SelfAdjointStrip}, quantization proceeds canonically, i.e., we define the one-particle Hilbert space
\[
 \HS_1 = L^2(\R^{d-1}) \otimes l^2(\N),
\]
the corresponding symmetric Fock space $\sF$ and the usual annihilation and creation operators $a_m(k)$, $a_m(k)^*$ fulfilling
\[
 [a_{m}(k), a_{m'}(k')^*] = \delta_{m m'} \delta(k-k').
\]
For $\mu > 0$, we define, for $F = (f, f|) \in \dom(\Delta^{-\frac{1}{4}})$ and $G \in \dom(\Delta^{\frac{1}{4}})$, the time zero fields as
\begin{align*}
 \phi_0(F) & = \sum_m \int \frac{\ud^{d-1} k}{\sqrt{2 \omega_{k,m}}} \left( \skal{\bar F}{\Phi_{k,m}} a_m(k) + \skal{\Phi_{k,m}}{F} a_m(k)^* \right), \\
 \pi_0(G) & = - i \sum_m \int \ud^{d-1} k \frac{\sqrt{\omega_{k,m}}}{\sqrt{2}} \left( \skal{\bar G}{\Phi_{k,m}} a_m(k) - \skal{\Phi_{k,m}}{G} a_m(k)^* \right).
\end{align*}
Using the completeness of the basis, it is straightforward to check that these fulfill the canonical equal time commutation relations, i.e.,
\begin{align*}
 [\phi_0(F), \phi_0(F')] & = 0, \\
 [\pi_0(G), \pi_0(G')] & = 0, \\
 [\phi_0(F), \pi_0(G)] & = i \skal{\bar F}{G}, 
\end{align*}
\cf the symplectic form \eqref{eq:SymplecticForm}. Interestingly, the time-zero fields may be restricted to the boundary: Inserting $F = (0, f|)$, one obtains, using \eqref{eq:phi_k_m_restricted}
\[
 \phi_0(0, f|) = \sum_m \int \frac{\ud^{d-1} k}{\sqrt{2 \omega_{k,m}}} d_m \left( \hat f|(-k) a_m(k) + \hat f|(k) a_m(k)^* \right).
\]
Due to the decay of the coefficients $d_m$, \cf Proposition~\ref{prop:SelfAdjointStrip}, this operator is well defined on a dense domain for $f| \in L^2(\del \Sigma)$. However, this is not the case of the momentum $\pi_0$.

Note that for $\mu = 0$ and $d=1$, the zero mode has to be treated separately, using position and momentum operators $q$, $p$, corresponding to the linear growth of the classical mode. For $\mu = 0$ and $d=2$, we have the infrared problems for the zero mode that are usually present in 1+1 space-time dimension.

When defining space-time fields, we have to decide on the appropriate test function space. In view of the definition of the time-zero fields, it seems natural to allow for test functions $F = (f, f|) \in \sS(M) \oplus \sS(\del M)$ and defining
\begin{multline}
\label{eq:phi}
 \phi(F) = \sum_m \int \ud x^0 \frac{\ud^{d-1} k}{\sqrt{2 \omega_{k,m}}} \left( \skal{\bar F(x^0)}{\Phi_{k,m}} e^{-i\omega_{k,m} x^0} a_m(k) \right. \\
 \left. + \skal{\Phi_{k,m}}{F(x^0)} e^{i\omega_{k,m} x^0} a_m(k)^* \right).
\end{multline}
These have the usual properties of Wightman fields (generalized to the present setting):

\begin{proposition}
\label{prop:Wightman}
Let $\mu > 0$.
The field $\phi(f, f|)$, with $(f, f|) \in \sS(M) \times \sS(\del M)$ real, is essentially self-adjoint on a dense invariant linear domain $\D \subset \sF$. For $\Omega_1, \Omega_2 \in \D$ the maps
\[
 \sS(M) \times \sS(\del M) \ni (f, f|) \mapsto \skal{\Omega_1}{\phi(f, f|) \Omega_2} \in \C
\]
are linear and continuous. The field $\phi$ is causal, i.e.,
\[
 [ \phi(f, f|), \phi(g, g|) ] = 0
\]
if the supports of $(f, f|)$ and $(g, g|)$ are space-like separated. There is a unitary representation $U$ of the proper orthochronous Poincar\'e group $\R^{1, d-1} \rtimes SO^+(1, d-1)$, under which the domain $\D$ is invariant and such that
\[
 U(a, \Lambda) \phi(f, f|) U(a, \Lambda)^* = \phi( f_{(a, \Lambda)}, f|_{(a, \Lambda)} )
\]
with
\begin{align*}
 f_{(a, \Lambda)}(x, z) & = f(\Lambda^{-1} (x - a), z), &
 f|_{(a, \Lambda)}(x) & = f|(\Lambda^{-1} (x - a)).
\end{align*}
The vacuum vector $\Omega \in \D$ is invariant under $U$, cyclic \wrt polynomials of the fields $\phi(f, f|_{\del \Sigma})$ or $\phi(0, f|)$, and the spectrum of $U$ on $\Omega^\perp$ is contained in
\[
 V_\mu = \left\{ p \in \R^{1, d-1} \mid - p^2 \geq \mu^2, p^0 \geq 0 \right\}.
\]
\end{proposition}

Parts of the proof consist in mapping to the generalized free field $\psi$ on $\R^{1, d-1}$ with momentum space weight
\[
 \sum_m \delta(-k^2 - q_m^2 - \mu^2) \ud^d k.
\]
Obviously, the Fock spaces $\sF$ and $\sF_\psi$ on which $\phi$ and $\psi$ act can be trivially identified by the isomorphism
\[
 \left( (\imath \tilde \Omega)_j \right)_{m_1 \dots m_j}(k_1, \dots, k_j) = \left( \tilde \Omega_j \right)_{m_1 \dots m_j}(k_1, \dots, k_j).
\]
With this identification, the field $\psi$ can be expressed as
\begin{equation}
\label{eq:psi}
 \psi(x) = (2 \pi)^{-\frac{d-1}{2}} \sum_m \int \frac{\ud^{d-1} k}{\sqrt{2 \omega_{k,m}}} \left( e^{-i\omega_{k,m} x^0 + i k x} a_m(k) + e^{i\omega_{k,m} x^0 - i k x} a_m(k)^* \right).
\end{equation}
For the mapping of the fields, we need the following lemma:
\begin{lemma}
\label{lemma:mapping}
Let $\mu > 0$ and $\hat f^\pm \in \sS(\R^{d-1}) \otimes s(\N)$, with $s(\N)$ the space of sequences that fall off faster than any power. There is a linear continuous mapping 
\[
 (\sS(\R^{d-1}) \otimes s(\N))^2 \ni (\hat f^+, \hat f^-) \mapsto \hat f' \in \sS(\R^{d})
\]
such that
\[
 \hat f^\pm_m(k) = \hat f'(\pm \omega_{k, m}, \pm k) \qquad \forall k \in \R^{d-1}, m \in \N,
\]
with $\omega_{k, m}$ given by \eqref{eq:omega_k_m}.
\end{lemma}

\begin{proof}
Choose $\chi: \R \to \R$, smooth and supported in $[-1/2,1/2]$, with $\chi(0) = 1$. Then choose $a > 0$ such that $\mu^2 < 1/a$ and $q_m^2 - q_{m-1}^2 < 1/a$ for all $m \geq 1$ (this is possible by Proposition~\ref{prop:SelfAdjointStrip}) and set
\[
  \hat f'(\omega, k) = \sum_m \sum_{s \in \pm} \theta(s \omega) \chi( a (\omega^2-\omega_{k,m}^2)) \hat f^s_m(s k).
\]
Clearly, this fulfills the requirement and is a Schwarz function. Furthermore, each seminorm
\[
 \norm{\hat f'}_{\alpha \beta} = \sup_{k \in \R^d} \betrag{k^\alpha} \betrag{ \del^\beta \hat f'(k)}
\]
with multiindices $\alpha$, $\beta$, can be bounded by the seminorms
\[
 \norm{ \hat f^\pm }_{\alpha \beta c} = \sum_{m \in \N} \sup_{k \in \R^{d-1}} m^c \betrag{k^\alpha} \betrag{ \del^\beta \hat f^\pm_m(k)}
\]
of $\sS(\R^{d-1}) \otimes s(\N)$, due to the growth of $\omega_{k, m}$ with $m$, \cf Proposition~\ref{prop:SelfAdjointStrip}.
\end{proof}

\begin{proof}[Proof of Proposition~\ref{prop:Wightman}]
Causality is a consequence of the canonical equal-time commutation relations and the causal propagation, proved in the previous section. The action of the Poincar\'e group is straightforwardly defined by its action on the one-particle space $\HS_1$,
\[
 (U(a, \Lambda) f)_m(\tilde k) = e^{i a \tilde k} f_m(\Lambda \tilde k),
\]
where $\tilde k = (\omega_{k,m}, k)$ and by an abuse of notation we identified the function $f_m(k)$ on $\R^{d-1}$ with the the function $f_m(\tilde k)$ on the mass hyperboloid for mass $\sqrt{q_m^2 + \mu^2}$. This also entails the claim on the spectrum.

To prove continuity, we define continuous linear maps
\[
 \sS(M) \times \sS(\del M) \ni (f, f|) \mapsto \hat f^\pm \in \sS(\R^{d-1}) \otimes s(\N)
\]
by
\begin{align*}
 \hat f^-_m(k) & = (2 \pi)^{- \frac{1}{2}} \int \skal{(\bar f, \bar f|)(x^0)}{\Phi_{k,m}} e^{-i \omega_{k, m} x^0} \ud x^0, \\
 \hat f^+_m(k) & = (2 \pi)^{- \frac{1}{2}} \int \skal{\Phi_{k,m}}{(f, f|)(x^0)} e^{i \omega_{k, m} x^0} \ud x^0.
\end{align*}
Continuity can be shown by using the growth of $\omega_{k, m}$ with $m$, Proposition~\ref{prop:SelfAdjointStrip}. Using the previous lemma, we can map this pair further to $\hat f' \in \sS(\R^d)$. By construction, we then have
\[
 \phi(f, f|) = \imath^{-1} \circ \psi(f') \circ \imath,
\]
with $\psi$ the generalized free field \eqref{eq:psi}, and $\imath$ the canonical isomorphism of the Fock spaces. For $\psi$, self-adjointness on the invariant domain
\[
 \D = \left\{ \tilde \Omega \in \bigoplus_{j = 0}^n (\sS(\R^{d-1}) \otimes s(\N))^{\otimes_s j} \mid n \in \N \right\},
\]
with $\otimes_s j$ the $j$ symmetric tensor power was shown in \cite[Section~II.6]{Jost}. 

Also the cyclicity of the vacuum \wrt polynomials of the field $\psi$ was proven in \cite[Section~II.6]{Jost}. In order to prove the cyclicity \wrt polynomials of $\phi(0, f|)$, it thus suffices to construct a map $\sS(\R^d) \ni f' \mapsto f| \in \sS(\del M)$ such that
\[
 \phi(0, f|) = \imath^{-1} \circ \psi(f') \circ \imath.
\]
We define
\[
 {\hat{f'}}_m^\pm(k) = \hat f'( \pm \omega_{k,m}, \pm k) d_m,
\]
with $d_m$ the constants defined in Proposition~\ref{prop:SelfAdjointStrip},
and then use Lemma~\ref{lemma:mapping} to obtain $\hat f$ such that
\[
 \hat f(\pm \omega_{k,m}, \pm k) = {\hat{f'}}^\pm_m(k) \qquad \forall k \in \R^{d-1}, m \in \N.
\]
Placing $f$ at the $\del_+ M$ component of the boundary yields the desired result. For cyclicity \wrt polynomials in $\phi(f, f|_{\del M})$, given $f'$, we define
\begin{multline*}
 f(x,z) = (2 \pi)^{-\frac{1}{2}} \sum_m \int \ud \omega \ud^{d-1} k \ \theta(\omega) \phi_{k, m}(\underline x,z) \chi(\omega^2 - \omega_{k, m}^2) \\ \times \left( \hat f'(-\omega_{k, m}, -k) e^{i \omega x^0} + \hat f'(\omega_{k, m}, k) e^{-i \omega x^0} \right)
\end{multline*}
where $\chi(0) = 1$ and the support of the test function $\chi$ is taken sufficiently small, so that there is no overlap of the various mass shells. This has the desired property and using the explicit form of $\phi_{k, m}$ derived in Proposition~\ref{prop:SelfAdjointStrip}, one can show that $f \in \sS(M)$.
\end{proof}

Proposition~\ref{prop:Wightman} states that the restriction of the field to the boundary is rather naturally possible. Hence, we define, for $f \in \sS(\del M)$,
\begin{equation}
\label{eq:def_phi_bd}
 \phi|(f) = \phi(0, c^{-1} f),
\end{equation}
where the factor $c^{-1}$ is introduced to cancel the factor $c$ from the scalar product. We may also restrict to the two boundary components separately. This yields
\begin{multline}
\label{eq:phi_bd}
 \phi|_\pm(x) = (2 \pi)^{-\frac{d-1}{2}} \sum_m (\pm)^m d_m \int \frac{\ud^{d-1} k}{\sqrt{2 \omega_{k,m}}} \times \\ \times \left( e^{-i(\omega_{k,m} x^0 - k \underline{x})} a_m(k) + e^{i(\omega_{k,m} x^0 - k \underline{x})} a_m(k)^* \right),
\end{multline}
i.e., a generalized free field. The spectrum contains all the mass shells for the masses $q_m$. Its two-point function is given by
\begin{equation}
\label{eq:2pt_bdy}
 \Delta_+|_\pm(x) = \sum_m \betrag{d_m}^2 \Delta_+^{\mu_m}(x),
\end{equation}
where we made use of translation invariance, $\mu_m = \sqrt{q_m^2+\mu^2}$ and $\Delta_+^m$ is the two-point function for the free field of mass $m$ on the $d$ dimensional Minkowski space. For $d=1$, this is
\[
 \Delta_+^\mu(x^0) = \frac{1}{2 \mu} e^{- i  \mu x^0}.
\]
By Proposition~\ref{prop:SelfAdjointStrip}, the coefficients $d_m$ are square summable, which implies that the two-point function has the same degree of short-distance singularity as the vacuum two-point function in $d$ space-time dimensions. This can be formalized by the concept of the scaling degree of a distribution, \cf \cite{BrunettiFredenhagenScalingDegree}. We then have:

\begin{proposition}
\label{prop:2pt}
Let $\mu > 0$ or $d>2$. Then $\Delta_+|_\pm$ is a tempered distribution. Its singular support is contained in $\{ x \in \R^d | x^2 \leq 0 \}$ and the projection of its analytic wave front set to the cotangent space is given by $\{ k \in \R^d | k^2 \leq 0, k^0 > 0 \}$. For $d \geq 2$, the scaling degree of $\Delta_+|_\pm$ at coinciding points is $d-2$.
\end{proposition}

\begin{proof}
We consider its Fourier transform
\[
 \hat \Delta_+|_\pm(k) = \sum_m \betrag{d_m}^2 \hat \Delta_+^{\mu_m}(k).
\]
For the individual
\[
  \hat \Delta_+^{\mu}(k) = (2\pi)^{-1} \theta(k^0) \delta(-k^2 - \mu^2)
\]
we can derive a bound
\[
 \betrag{\skal{\hat \Delta_+^\mu}{\phi}} \leq C \norm{\phi}_{d},
\]
with the seminorm
\[
 \norm{\phi}_{d} = \sup_{k \in \R^d} (1 + \betrag{k})^{d} \betrag{ \phi(k) }
\]
of the space of Schwarz distributions.
This proceeds as follows:
\begin{align*}
 \betrag{\skal{\hat \Delta_+^\mu}{\phi}} & \leq (2\pi)^{-1} \int \frac{\ud^{d-1} k}{2 \omega_\mu(k)} \betrag{\phi(\omega_\mu(k), k)} \\
 & \leq (2\pi)^{-1} \norm{\phi}_{d} \int \frac{\ud^{d-1} k}{2 \omega_\mu(k)} \frac{1}{( 1 + \sqrt{2 k^2 + \mu^2 })^{d}}.
\end{align*}
For the latter integral, one can derive a bound which is independent of $\mu$. Temperedness then follows from the square summability of the $d_m$. For the singular support, consider any compact set $K$ space-like to the origin. There, $\Delta_+|_\pm$ is given by
\[
 \Delta_+|_\pm(x) = \sum_m \betrag{d_m}^2 \frac{1}{(2 \pi)^{-d/2}} \mu_m^{d/2-1} \betrag{x^2}^{1/2-d/4} K_{d/2-1}(\sqrt{\mu_m^2 x^2})
\]
Due to \cite[9.6.29]{AbramowitzStegun}
\[
 \del_z K_\nu(z) = - \tfrac{1}{2} \left( K_{\nu-1}(z) + K_{\nu+1}(z) \right)
\]
and the exponential decay of $K_\nu(z)$ for $z \to \infty$, all derivatives of $\Delta_+|_\pm$ can be uniformly bounded on $K$, so that the series converges to a smooth function outside of the light cone. That only positive frequency momenta are contained in the analytic wave front set follows from the support properties of $\hat \Delta_+|_\pm$. That positive frequency momenta with $k^2 < 0$ are contained follows from the fact that the decay in momentum space is not faster than any power, due to the tower of mass shells with a weight that only decays as a power law. For the scaling degree, one may use the asymptotic form of $K_\nu(z)$ for $z \to 0$ and the square-summability of the $d_m$.
\end{proof}

For time-like separations, one expects supplementary singularities coming from reflections at the other boundary.

Obviously, the representation of $\phi|_+$ on the Fock space $\sF$ coincides with the GNS representation of the Borchers-Uhlmann algebra for fields on $\R^{1, d-1}$ \wrt the two-point function \eqref{eq:2pt_bdy}. As a consequence of the above bound on the analytic wave front set and \cite[Cor.~5.5]{StrohmaierVerchWollenberg02}, the boundary field $\phi|_+$ fulfills the Reeh-Schlieder property, 
i.e., the set of states obtained by acting with the bounded operators localized in a fixed, arbitrarily small region of the boundary on the vacuum is dense in $\sF$,
\cf \cite{StrohmaierVerchWollenberg02} for details.

In analogy with \eqref{eq:def_phi_bd}, we may also define the bulk field as
\[
 \phi_\bk(f) = \phi(f,0)
\]
for $f \in \sS(M)$. Admitting more singular smearing functions $f$, we readily arrive at
\begin{equation}
\label{eq:phi_bd_phi_bk}
 \phi|_\pm(f) = \phi_\bk(f \delta(z \mp S)).
\end{equation}
Furthermore, it is straightforward to check that the boundary value of $\del_\perp \phi_\bk$ is indeed the source for the boundary field $\phi|$:
\begin{equation}
\label{eq:phi_bd_source}
 \phi|_\pm((- \Box + \mu^2) f) = \mp c^{-1} \phi_\bk(f \delta'(z \mp S)).
\end{equation}

It is straightforward to define Wick powers of the bulk and boundary fields, either by a coinciding point limit with subtraction of the correct combination of lower order Wick powers and two-point functions (dictated by Wick's theorem), or by normal ordering the creation and annihilation operators. In any case, the boundary Wick power is given by the boundary value of the bulk Wick power, i.e.,
\[
 \phi|_{\pm}^k(f) = \phi_\bk^k(f \delta(z \mp S)).
\]
\begin{remark}
Instead of defining Wick powers globally by normal ordering, it may be more appropriate to choose a local prescription, as advocated in the context of quantum field theory on curved space-times \cite{HollandsWaldReview}. However, problems may then occur at the boundary, \cf \cite{DappiaggiNosariPinamonti14}, for example, i.e., it may be necessary to restrict to bulk test functions supported in the interior of the bulk. Also the treatment of the boundary part needs to be clarified then.
\end{remark}

For each mode $m$, one may also define local fields $\phi_m|_\pm$, localized at the boundary:
\[
 \phi_m|_\pm(x) = \frac{(\pm)^m}{(2 \pi)^{\frac{d-1}{2}}} \int \frac{\ud^{d-1} k}{\sqrt{2 \omega_{k,m}}} \left( e^{-i(\omega_{k,m} x^0 - k \underline{x})} a_m(k) + e^{i(\omega_{k,m} x^0 - k \underline{x})} a_m(k)^* \right).
\]
The relation between the fields $\phi_m|_\pm$ and $\phi|_\pm$ is clearly non-local, i.e., for
\begin{equation}
\label{eq:ModeFieldRepresentation}
 \phi_m|_\pm(f_m) = \phi|_\pm(f)
\end{equation}
to hold, the test function $f$ must be de-localized. This can be made quite precise:
\begin{proposition}
\label{prop:ModeFieldRepresentation}
Let $d=1$. For \eqref{eq:ModeFieldRepresentation} to hold, $f$ can not be supported on an interval smaller than $\frac{8}{\pi} S$.
\end{proposition}

\begin{proof}
Assume $f$ is localized in an interval of length $L$. By the Paley-Wiener theorem \cite{ReedSimonI}, its Fourier transform is an entire function which is bounded by
\[
 \betrag{\hat f(\xi)} \leq C e^{\frac{L}{2} \betrag{\Im(\xi)}}
\]
for some constant $C$, i.e., it is of order $1$ and type $\tau \leq \frac{L}{2}$ \cite{Boas}. On the other hand, in order for \eqref{eq:ModeFieldRepresentation} to hold, we must have $\hat f(\pm q_{m'}) = 0$ for all $m' \neq m$. From Proposition~\ref{prop:SelfAdjointStrip} and \cite[Thm.~9.1.4]{Boas} it follows that $\tau > \frac{4}{\pi} S$ unless $f$ vanishes.
\end{proof}

From causality, i.e., local commutativity, it is clear that the boundary field does not fulfill the time-slice axiom: All boundary observables contained in a small time-slice commute with space-like separated bulk observables. 
It follows that the time-slice axiom does not hold for time slices smaller than $2 S$, \cf \cite{HaagSchroer62} for the discussion for a generalized free field with continuous K\"allen-Lehmann weight. Whether the time-slice axiom holds for time slices larger than $2 S$ is however not clear. The previous proposition suggest that even in the case $d = 1$ the minimal time slice should be at least $\frac{8}{\pi} S$.

We now want to construct the holographic map, i.e., for a given bulk field $\phi_\bk(f)$ we want to find a test function $f'$ on\footnote{For simplicity, we restrict to the right boundary. One may of course also consider $\del_- M$ or $\del M$.} $\del_+ M$ such that
\begin{equation}
\label{eq:Holography}
 \phi_\bk(f) = \phi|_+(f').
\end{equation}
We denote
\begin{align*}
 \hat f_m^-(k) & = (2 \pi)^{- \frac{1}{2}} \int \ud x^0 \ \skal{(\bar f(x^0), 0)}{\Phi_{k, m}} e^{- i \omega_{k, m} x^0}, \\
 \hat f_m^+(k) & = (2 \pi)^{- \frac{1}{2}} \int \ud x^0 \ \skal{\Phi_{k, m}}{(f(x^0), 0)} e^{i \omega_{k, m} x^0}.
\end{align*}
By Proposition~\ref{prop:SelfAdjointStrip}, we have to find $f' \in \sS(\del_+ M)$ such that
\[
 \hat f^\pm_m(k) = d_m \hat f'(\pm \omega_{k, m}, \pm k) \qquad \forall k \in \R^{d-1}, m \in \N.
\]
As $d_m \neq 0$ and $\{ d_m^{-1} \}$ is polynomially bounded, \cf Proposition~\ref{prop:SelfAdjointStrip}, there is an $f' \in \sS(\del_+ M)$ with this property, \cf Lemma~\ref{lemma:mapping}. We have thus proven:
  
\begin{proposition}
Let $\mu^2 > 0$. Then to each $f \in \sS(M)$ there exists $f' \in \mathcal{S}(\del_+ M)$ such that \eqref{eq:Holography} with $\phi|_\pm$ defined by \eqref{eq:phi_bd} holds.
\end{proposition}

An interesting open question is whether $f'$ can be chosen to be compactly supported if $f$ is. By the Paley-Wiener theorem, a necessary condition for this is that $\hat f'$ can be chosen to be an entire function of exponential type. For $d=1$, this is possible due to Proposition~\ref{prop:SelfAdjointStrip} and \cite[Thm.~2]{Leontev58}.\footnote{I am grateful to Michael Bordag for helping with the translation of \cite{Leontev58}.} But it is not clear whether one can also achieve the necessary fall-off in the real direction. In any case, one can not expect this to work in higher dimensions, due to the non-analyticity of $\omega_{k, m}$ in $k$. This is in contrast to the case of holography on AdS, where a localized bulk observable can always be mapped to a localized boundary observable \cite{RehrenHolography}.

A similar mapping can be constructed for Wick powers. However, one can not expect that one obtains local Wick powers at the boundary. Instead, a function $f' \in \mathcal{S}(\del_+M^k)$ such that
\[
 \phi_\bk^k(f) = \int \ud^{d} x_1 \dots \ud^d x_k \ : \! \phi|_+(x_1) \dots \phi|_+(x_k) \! : f'(x_1, \dots, x_k)
\]
can easily be constructed along the lines discussed above.

For illustration, we can consider how a holographic image of a local observable actually looks like. We choose $d=1$, $S=1$, $c = 1$, $\mu = 0$ and
\begin{equation}
\label{eq:fTest}
 f(t, x) = \begin{cases} e^{-\frac{1}{t+1/2}} e^{-\frac{1}{1/2-t}} e^{-\frac{1}{x+1/2}} e^{-\frac{1}{1/2-x}} & x,t \in (-1/2, 1/2) \\
0 & \text{otherwise}. \end{cases}
\end{equation}
A holographic dual $f'$ to this is shown in Figure~\ref{fig:B2B}. We see a sequence of oscillations, around $t= \pm 1, \pm 3, \pm 5$. These correspond to the times when the propagation of the peak of the test function hits the boundary after $0, 1, 2$ reflections at the boundary.

\begin{figure}
\begin{center}
\includegraphics[width=0.8\textwidth]{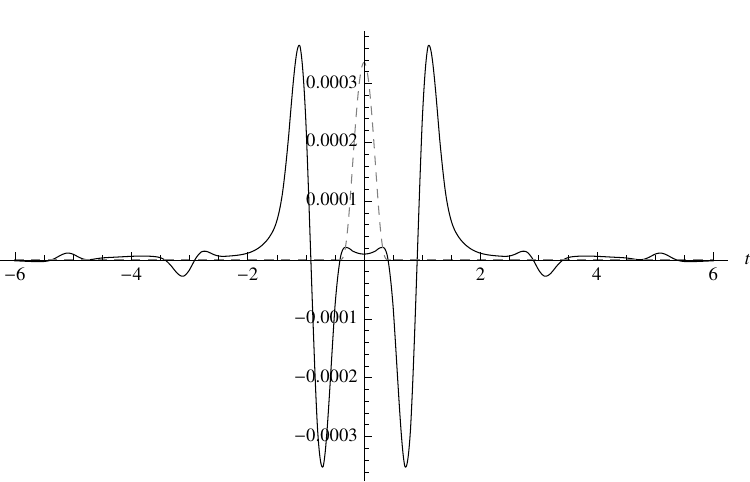}
\caption{A holographic image $f'(t)$ of \eqref{eq:fTest}. The dashed curve shows $f(t, 0)$.}
\label{fig:B2B}
\end{center}
\end{figure}

Let us comment on the situation for the half-space $\Sigma = \R_+^d$. We then have the orthonormal basis $\{ \Phi_{k,q} \}_{k \in \R^{d-1}, q \in \R_+}$, \cf Proposition~\ref{prop:SelfAdjoint}, so that the one-particle Hilbert space is given by
\[
 \HS_1 = L^2(\R^{d}_+),
\]
and the annihilation and creation operators $a(k,q)$, $a(k,q)^*$ fulfilling
\[
 [a(k,q), a(k',q')^*] = \delta(k-k') \delta(q-q').
\]
The local field is then defined as
\begin{multline*}
 \phi(F) = \int \ud x^0 \frac{\ud^{d-1} k}{\sqrt{2 \omega_{k,m}}} \ud q \left( \skal{\bar F(x^0)}{\Phi_{k,q}} e^{-i\omega_{q,m} x^0} a(k,q) \right. \\
 \left. + \skal{\Phi_{k,q}}{F(x^0)} e^{i\omega_{k,q} x^0} a(k,q)^* \right),
\end{multline*}
for $F = (f, f|)$.
Note that for $\mu = 0$ and $d=1$, we have the infrared problems that are usually present in 1+1 space-time dimension.

As for the case of the strip, the fields can be restricted to the boundary, yielding a generalized free field with the two-point function
\[
 \Delta_+ = \int_0^\infty \ud q \frac{2}{\pi ( c^2 q^2 + 1)} \Delta_+^{\sqrt{\mu^2 + q^2}}.
\]
This is again a tempered distribution with scaling degree $d-2$ for $d \geq 2$. A holographic map $\sS(M) \ni f \mapsto f'$ such that
\[
 \phi_\bk(f) = \phi|(f')
\]
can be defined as follows. Define
\begin{align*}
 \hat f^-(k, q) & = (2 \pi)^{-\frac{1}{2}} \int \ud x^0 \ \skal{(\bar f(x^0), 0)}{\Phi_{k, q}} e^{- i \omega_{k, q} x^0}, \\
 \hat f^+(k, q) & = (2 \pi)^{-\frac{1}{2}} \int \ud x^0 \ \skal{\Phi_{k, q}}{(f(x^0), 0)} e^{i \omega_{k, q} x^0},
\end{align*}
and set
\[
 \hat f'(\omega, k) = \sqrt{ \frac{\pi (c^2 q^2 + 1)}{2}}
  \begin{cases}
  \hat f^+(k, q), & \omega^2 - k^2 > \mu^2, \omega > 0, \\
  \hat f^-(-k, q), & \omega^2 - k^2 > \mu^2, \omega < 0.
 \end{cases}
\]
with $q = \sqrt{\omega^2 - k^2 - \mu^2}$. The problem is that this can in general not be smoothly continued to the region $\omega^2 - k^2 \geq \mu^2$. Hence, the resulting $f'$ will in general be a smooth $L^2$ function. In the half-space setting, the holographic map is thus more delocalizing than in the strip case.

Also thermodynamically, the boundary field in the strip-space case is less appealing, as it does not fulfill Buchholz-Wichmann nuclearity \cite{BuchholzWichmannNuclearity}. Nevertheless, for $\mu > 0$, one can straightforwardly define KMS states.

\section{Relation to other boundary conditions and the AdS/CFT correspondence}
\label{sec:Comments}

First of all, let us compare to more common boundary conditions. For concreteness, let us consider Neumann boundary conditions. Restriction of time-zero fields to the boundary is then certainly not possible. But when smeared in time, the bulk field may be restricted to the boundary, so that \eqref{eq:phi_bd_phi_bk} holds. 
However, the resulting theory has bad short distance behavior: It has the same degree of singularity as the original $d+1$ dimensional field theory, but only lives in $d$ space-time dimensions. This is due to the fact that for Neumann boundary conditions the coefficients $\{ d_m \}$ would be constant. Apart from the statement about the degree of singularity, all properties listed in Proposition~\ref{prop:2pt} also hold for the two-function of boundary fields in Neumann boundary conditions. Obviously, \eqref{eq:phi_bd_source} does not hold for Neumann boundary conditions, but the source on the \rhs would be given by $\del_\perp^2 \phi_\bk|_{\del M}$.

One could also consider Dirichlet boundary conditions, but then the restriction to the boundary would be trivial. However, one may restrict $\del_\perp \phi$ to the boundary, yielding a field theory on the boundary whose singular behavior is even worse than in the case of Neumann boundary conditions.

The Dirichlet case is quite similar to what happens in the AdS/CFT correspondence as applied to massive scalar fields \cite{WittenHolography}. There, one has two bulk propagators $G_\pm$, whose leading behavior near the boundary $z= 0$ is $z^{\Delta_\pm}$ with $\Delta_+ > \Delta_-$. Hence, one may understand $G_+$ as the analog of the Dirichlet propagator and $G_-$ as the analog of the Neumann propagator on flat space (where we would have $\Delta_+ = 1$ and $\Delta_- = 0$). The dual field can be understood as the limit $\mathcal{O}_\phi(x) = \lim_{z \to 0} z^{- \Delta_+} \phi(x, z)$ \cite{DutschRehrenDualField}, whose analog in flat space is $\del_\perp \phi$ restricted to the boundary. Its source $\phi_0$ is interpreted as the boundary value of the bulk field in the sense that $\phi \sim z^{\Delta_-} \phi_0$. It's flat space analog is thus the restriction to the boundary. Hence, in Wentzell boundary conditions, the analog the dual field is the source for the analog of the boundary value, in the sense of \eqref{eq:phi_bd_source}. However, both fields are quantized in that setting.

In holographic renormalization \cite{Skenderis02}, one introduces, in suitable coordinates, a boundary at $z = \eps$ and considers the limit $\eps \to 0$. Now the classical action, evaluated on a solution with the boundary value $\phi_0$ in the above sense, diverges as $\eps \to 0$. The cure is to introduce $\eps$ dependent counterterms localized on the boundary. For the massive scalar field, these boundary counterterms are of the type considered here, but with $\mu_\bk \neq \mu_\bd$ and $c < 0$. As discussed in Section~\ref{sec:WaveEquation}, negative values of $c$ lead to severe difficulties already at the classical level. But as holographic renormalization is a formal mathematical trick, it is not clear whether this poses real problems. In any case, working formally with this negative value of $c$, one finds that for finite $\eps$, the singularities of the boundary field are indeed of the expected form, i.e., the two-point function is singular of degree $d-2$, \cf also the discussion in Remark~\ref{rem:BoundaryBehaviorGeneral}.

\section{Conclusion}

We established well-posedness of the Cauchy problem for \eqref{eq:eom_bk}, \eqref{eq:eom_bd}, including causality. We also quantized the system and discussed some properties of the resulting quantum field. In particular, we showed that the field may be restricted to the boundary, yielding a generalized free field with the degree of singularity that one would expect for a scalar field in $d$ dimensions. Finally, we constructed an explicit holographic correspondence between bulk and boundary fields.

Our results lead to a couple of new questions. Regarding generalized Wentzell boundary conditions, a proper microlocal calculus should be set up, which would allow to prove propagation of singularities and thus give some control over the wave front sets of the relevant propagators.

One could also work out generalized Wentzell boundary conditions for other fields. In particular, it would be interesting to see whether they can be defined for Dirac fields. This may involve chiral fields on the boundaries.
For gauge fields, there naturally appears a supplementary scalar on the boundaries, the boundary value of the normal component of the bulk vector potential. Hence, such models are potentially interesting also phenomenologically.

Regarding the holographic aspect, it would be interesting to investigate whether the holographic relations still hold for interacting theories.

\subsection*{Acknowledgments}
I would like to thank Claudio Dappiaggi, Stefan Hollands, Karl-Henning Rehren, Ko Sanders, Rainer Verch, Ingo Witt, and Michal Wrochna for helpful discussions or remarks. This is a post-peer-review, pre-copyedit version of an article published in Annales Henri Poincar\'e. The final authenticated version is available online at: https://doi.org/10.1007/s00023-017-0629-3


\begin{thebibliography}{10}

\bibitem{tHooftHolography}
G. 't~Hooft,
\newblock {Dimensional reduction in quantum gravity}, (1993),
  [arXiv:gr-qc/9310026].

\bibitem{SusskindHolography}
L. Susskind,
\newblock {The World as a hologram}, J.\ Math.\ Phys. 36 (1995) 6377,
  [arXiv:hep-th/9409089].

\bibitem{MaldacenaHolography}
J.M. Maldacena,
\newblock {The Large N limit of superconformal field theories and
  supergravity}, Adv.\ Theor.\ Math.\ Phys. 2 (1998) 231,
  [arXiv:hep-th/9711200].

\bibitem{WittenHolography}
E. Witten,
\newblock {Anti-de Sitter space and holography}, Adv.\ Theor.\ Math.\ Phys. 2
  (1998) 253, [arXiv:hep-th/9802150].

\bibitem{RehrenHolography}
K.H. Rehren,
\newblock {Algebraic holography}, Annales Henri Poincar{\'e} 1 (2000) 607,
  [arXiv:hep-th/9905179].

\bibitem{BertolaBrosMoschellaSchaeffer00}
M. Bertola et~al.,
\newblock {A general construction of conformal field theories from scalar
  anti-de Sitter quantum field theories}, Nucl.\ Phys. B587 (2000) 619.

\bibitem{DutschRehrenDualField}
M. D{\"u}tsch and K.H. Rehren,
\newblock {A Comment on the dual field in the scalar AdS / CFT correspondence},
  Lett.\ Math.\ Phys. 62 (2002) 171, [arXiv:hep-th/0204123].

\bibitem{ChodosThorn74}
A. Chodos and C.B. Thorn,
\newblock {Making the Massless String Massive}, Nucl.\ Phys. B72 (1974) 509.

\bibitem{RotatingString}
J. Zahn,
\newblock {The excitation spectrum of rotating strings with masses at the
  ends}, JHEP 1312 (2013) 047, [arXiv:1310.0253].

\bibitem{BalasubramanianKraus99}
V. Balasubramanian and P. Kraus,
\newblock {A Stress tensor for Anti-de Sitter gravity}, Commun.\ Math.\ Phys.
  208 (1999) 413, [arXiv:hep-th/9902121].

\bibitem{Skenderis02}
K. Skenderis,
\newblock {Lecture notes on holographic renormalization}, Class.\ Quant.\ Grav.
  19 (2002) 5849, [arXiv:hep-th/0209067].

\bibitem{Symanzik81}
K. Symanzik,
\newblock {Schr\"odinger Representation and Casimir Effect in Renormalizable
  Quantum Field Theory}, Nucl.\ Phys. B190 (1981) 1.

\bibitem{Ueno73}
T. Ueno,
\newblock Wave equation with {W}entzell's boundary condition and a related
  semigroup on the boundary. {I}, Proc. Japan Acad. 49 (1973) 672.

\bibitem{FaviniGoldstein2Romanelli02}
A. Favini et~al.,
\newblock The heat equation with generalized wentzell boundary condition, J.\
  Evol.\ Equ. 2 (2002) 1.

\bibitem{GalGoldstein203}
C.G. Gal, G.R. Goldstein and J.A. Goldstein,
\newblock Oscillatory boundary conditions for acoustic wave equations, J. Evol.
  Equ. 3 (2003) 623.

\bibitem{CocliteFaviniGoldstein2Romanelli14}
G.M. {Coclite} et~al.,
\newblock {Continuous dependence in hyperbolic problems with Wentzell boundary
  conditions.}, {Commun.\ Pure Appl.\ Anal.} 13 (2014) 419.

\bibitem{Vitillaro13}
E. Vitillaro,
\newblock Strong solutions for the wave equation with a kinetic boundary
  condition,
\newblock Recent trends in nonlinear partial differential equations. {I}.
  {E}volution problems, Contemp. Math. Vol. 594, pp. 295--307, Amer. Math.
  Soc., Providence, RI, 2013.

\bibitem{Vitillaro15}
E. Vitillaro,
\newblock {On the the wave equation with hyperbolic dynamical boundary
  conditions, interior and boundary damping and source}, Arch.\ Rational Mech.\ Anal. 223 (2017) 1183,
  [arXiv:1506.00910].

\bibitem{Feller57}
W. Feller,
\newblock Generalized second order differential operators and their lateral
  conditions, Illinois J.\ Math. 1 (1957) 459.

\bibitem{Greenberg}
O. Greenberg,
\newblock {Generalized Free Fields and Models of Local Field Theory}, Annals
  Phys. 16 (1961) 158.

\bibitem{TaylorPDEsI}
M.E. Taylor,
\newblock Partial differential equations {I}. {B}asic theory, Applied
  Mathematical Sciences Vol. 115, Second ed. (Springer, New York, 2011).

\bibitem{WaldGR}
R.M. Wald,
\newblock General Relativity (University of Chicago Press, 1984).

%

\bibitem{Jost}
R. Jost,
\newblock {The General Theory of Quantized Fields} (American Mathematical
  Society, 1965).

\bibitem{BrunettiFredenhagenScalingDegree}
R. Brunetti and K. Fredenhagen,
\newblock {Microlocal analysis and interacting quantum field theories:
  Renormalization on physical backgrounds}, Commun.\ Math.\ Phys. 208 (2000)
  623, [arXiv:math-ph/9903028].

\bibitem{AbramowitzStegun}
M. Abramowitz and I.A. Stegun,
\newblock Handbook of mathematical functions (Dover, New York, 1972).

\bibitem{StrohmaierVerchWollenberg02}
A. Strohmaier, R. Verch and M. Wollenberg,
\newblock {Microlocal analysis of quantum fields on curved space-times:
  Analytic wavefront sets and Reeh-Schlieder theorems}, J.\ Math.\ Phys. 43
  (2002) 5514, [arXiv:math-ph/0202003].

\bibitem{HollandsWaldReview}
S. Hollands and R.M. Wald,
\newblock {Quantum fields in curved spacetime}, Phys.\ Rept. 574 (2015) 1,
  [arXiv:1401.2026].

\bibitem{DappiaggiNosariPinamonti14}
C. Dappiaggi, G. Nosari and N. Pinamonti,
\newblock {The Casimir effect from the point of view of algebraic quantum field
  theory}, Math.\ Phys.\ Anal.\ Geom. 19 (2016) 12, [arXiv:1412.1409].

\bibitem{ReedSimonI}
M. Reed and B. Simon,
\newblock Methods of modern mathematical physics. {I}, Second ed. (Academic
  Press, Inc., New York, 1980).

\bibitem{Boas}
R.P. Boas, Jr.,
\newblock Entire functions (Academic Press Inc., New York, 1954).

\bibitem{HaagSchroer62}
R. Haag and B. Schroer,
\newblock {Postulates of Quantum Field Theory}, J.\ Math.\ Phys. 3 (1962) 248.

\bibitem{Leontev58}
A.F. Leont{'}ev,
\newblock Values of an entire function of finite order at given points, Izv.
  Akad. Nauk SSSR Ser. Mat. 22 (1958) 387.

\bibitem{BuchholzWichmannNuclearity}
D. Buchholz and E.H. Wichmann,
\newblock {Causal Independence and the Energy Level Density of States in Local
  Quantum Field Theory}, Commun.Math.Phys. 106 (1986) 321.

\end{thebibliography}

\end{document}